\documentclass[runningheads]{llncs}
\usepackage{graphicx}

\usepackage{times}
\usepackage{url}
\usepackage[hidelinks]{hyperref}
\usepackage[utf8]{inputenc}
\usepackage[small]{caption}
\usepackage{amsmath}
\usepackage{amsfonts}
\usepackage{booktabs}
\usepackage[ruled,vlined,linesnumbered,noresetcount]{algorithm2e}
\usepackage[dvipsnames]{xcolor}
\usepackage{bbm}
\usepackage[normalem]{ulem}
\usepackage{mathrsfs}
\usepackage{float}
\usepackage{eufrak}
\usepackage{stmaryrd}
\usepackage{wasysym}
\usepackage{array}
\usepackage{bigstrut,multirow}

\newtheorem{mechanism}{Mechanism}

\newcommand{\hau}[1]{\textcolor{OliveGreen}{Hau: \textrm{#1}}}

\DeclareMathAlphabet{\mathpzc}{OT1}{pzc}{m}{it}

\begin{document}
\title{Facility Location Games with Ordinal Preferences
}

%
\author{Hau Chan\inst{1} \and
Minming Li\inst{2}\and
Chenhao Wang\inst{1}}
\authorrunning{Chan et al.}
%
\institute{University of Nebraska-Lincoln, NE, USA
\email{\{hchan3,cwang55\}@unl.edu}   \and
City University of Hong Kong, HKSAR, China
\email{minming.li@cityu.edu.hk}
}
\maketitle              
\begin{abstract}
 We consider a new setting of facility location games with ordinal preferences. In such a setting, we have a set of agents and a set of facilities. Each agent is located on a line and has an ordinal preference over the facilities. Our goal is to design strategyproof mechanisms that elicit truthful information (preferences and/or locations) from the agents and locate the facilities to minimize both maximum and total cost objectives as well as to maximize both minimum and total utility objectives. For the four possible objectives, we consider the 2-facility settings in which only preferences are private, or locations are private.
 For each possible combination of the objectives and settings, we provide lower and upper bounds on the approximation ratios of strategyproof mechanisms, which are asymptotically tight up to a constant. Finally, we discuss the generalization of our results beyond two facilities and when the agents can misreport both locations and preferences.
\keywords{Facility location  \and Mechanism design \and Approximation.}
\end{abstract}

\section{Introduction}
Facility location games have been widely studied  in recent decades (see e.g., \cite{Dokow:2012aa,lu2010asymptotically,procaccia2009approximate}).
In the typical setting, we have a set of agents, a set of facilities,
and a set of possible locations (e.g., a bounded interval $[0, 1]$).
Each agent is located within the set of locations, where the agent's location is private information.
The goal is to design a \emph{strategyproof} mechanism that elicits true locations of the agents
and locates the facilities to (approximately) minimize the \emph{maximum}
distance/cost objective or \emph{total} distance/cost objective of the agents to their respective closest facilities.
Under this setting, each agent is indifferent about the facilities,
and, naturally, the agent's only interest is his/her closest facility.

When each agent has an \emph{ordinal} (or a complete ranking) preference over the facilities,
the standard facility location games and their variants \cite{Anastasiadis:2018aa,yuan2016facility} no longer
capture the tradeoff between the agent's ordinal preference and the agent's distances to the facilities.
Namely, even if a facility is closest to the agent, the agent could prefer
going to another facility that is farther away {due to} the agent's underlying ordinal preference.
In this work, our focus is to consider
facility location games that incorporate agents' ordinal preferences over the set of facilities,
{and design mechanisms} 
under several objectives and truthful elicitation requirements (i.e., preferences and/or locations).
Below, we provide several motivation examples in public facility domains for the necessity of considering
ordinal preferences in facility location games.

\smallskip\noindent\textbf{Motivation.}
{As a first example}, a planner wants to build two substitutable public facilities (e.g., schools, libraries, and parks)
on the real line (e.g., road), on which finitely many agents locate.
For instance, one can be a public library and the other can be a specialized library
(e.g., with special/unique contents or features). 
Both of {them} can serve and accommodate any agent in the system.
Some agents prefer the full range of general books in the public library,
while others prefer the specialized books and environment in the specialized library.
As a result, each agent has a preference over these two facilities.

{In the second example,} the planner plans to open two public schools (e.g., kindergartens)
with different education systems
(e.g., traditional, charter, montessori, or magnet) within a set of locations,
where one school ({e.g.}, magnet) pays more attention to some specific field of study,
and the other ({e.g.}, montessori) focuses on openness and guidance.
The parents have their own perspectives on education for their children
and therefore have different preferences over the two public schools.

In {both} of the above examples, we must consider
the tradeoff between the agent's ordinal preference
(first choice and second choice) and the agent's distances to the facilities.


\smallskip\noindent\textbf{Our Contribution.}
We introduce a new model of facility location games with ordinal preferences
where each agent can express his/her preference over the set of facilities
as well as his/her location (in the interval of [0,1]). {This model builds upon the facility location problems by incorporating agents’
preferences of facilities nontrivially.
Such a notion that allows the agents to
express their preferences over the facilities (such as preference over heterogeneous
parks and libraries) directly.
Adding such a preference feature requires us to nontrivially derive the appropriate facility location models.} 
We then propose appropriate cost objective and utility objective to capture
the tradeoff between an agent's preference and the agent's distance to the facilities.
The agent's cost (utility) for a facility is
his/her distance (one minus his/her distance) to the facility multiplied (divided) by a discount factor that
depends on the ranking in the agent's ordinal preference.
Naturally, an agent's cost (utility) objective is defined to be the minimum cost (maximum utility) over all facilities.
The interpretation is that less preferred facilities would be more costly
for the agents even if they are closer to the agents.

We study the problems of designing strategyproof and (approximately)
optimal {mechanisms} to minimize the maximum cost objective and total cost objective
as well as to maximize the minimum utility objective and total utility objective over all of the agents,
in settings where either ordinal preference or location information is private. Let $\alpha\ge 1$ be a constant  coefficient that characterizes the different preference of agents over the facilities, which can be viewed as a discount factor and will be formally defined in Section \ref{sec:model}.\footnote{For example, when there are two facilities, each agent incurs a cost equal to the minimum between
the distance to the second choice multiplied by factor $\alpha$, and
the distance to the first choice. As such, the role of $\alpha$ here is
to model the tradeoff between the distances and agent preferences (i.e., a less preferred facility will be viewed as further
away from the agent discounted by $\alpha$).}
{We leverage simple mechanisms to derive the upper bound (UB) results
and, to derive the lower bound (LB) results,
we carefully identify and construct various instances for the corresponding settings.}
Table \ref{tab:1} and \ref{tab:2} summarize our LB and UB results on the approximation ratios
of strategyproof mechanisms for the 2-facility setting, which are asymptotically  tight up to a small constant. {In particular, for the minimum utility objective in Table \ref{tab:1} with $\alpha\ge 2$, the bounds are exactly tight, and for the total cost objective in Table \ref{tab:2}, the bounds are very close when $\alpha$ is small.}
In Section \ref{sec:ex}, we discuss the generalization of our results
beyond two facilities and settings where preferences and locations are private.
We also discuss an alternative additive model where an agent's cost/utility for a facility is adjusted by adding/subtracting a constant based on the agent's ordinal preference. {All omitted proofs can be found in Appendix.}


\vspace{-8mm}\begin{table}[htbp]
  \centering
  \caption{A summary of our results with private preferences.}\label{tab:1}

    \begin{tabular}{|p{1.55cm}<{\centering}|p{4.8cm}<{\centering}|p{3.3cm}<{\centering}|}
    \hline
    \small Objective &  \small {Maximum/Minimum} &\small {Total} \bigstrut\\
    \hline
    \multirow{2}[2]{*}{\small Cost} & \multirow{1}[1]{*}{\small UB: $\alpha$  } & {\small UB: $\alpha$ } \bigstrut[t]\\

                                    & {\small LB: ~$\max\{1,\frac{\alpha}{2}\}$ } & {\small LB:~$\max\{1,\frac{\alpha}{4}\}$  } \bigstrut[b]\\
    \hline
    \multirow{2}[2]{*}{\small{Utility}} & {\small UB: 1 for $\alpha\ge2$, $\alpha$ for $\alpha< 2$ } & {\small UB: $\min\{2,\alpha\}$} \bigstrut[t]\\

                                    &  \multirow{2}[1]{*}{\shortstack{\small LB: 1 for $\alpha\ge 2$,  \\\small $\min\{\frac{2}{\alpha},\frac{3\alpha+1}{2\alpha+2}\}$ for $\alpha<2$}}
                                   & {\small  LB: $\frac{30\alpha^2+38\alpha+2}{30\alpha^2+37\alpha+3}$ ~~~~~~~~~~~~~~~} \bigstrut[b]\\
\hline
    \end{tabular}
\end{table}

\vspace{-15mm} \begin{table}[htbp]
  \centering
  \caption{A summary of our results with private locations.}
  \label{tab:2}

    \begin{tabular}{|p{1.55cm}<{\centering}|p{4.8cm}<{\centering}|p{3.3cm}<{\centering}|}
    \hline
    \small Objective &  \small {Maximum/Minimum} &\small {Total} \bigstrut\\
    \hline
    \multirow{2}[2]{*}{\small Cost} & \multirow{1}[1]{*}{\small UB: $2\alpha$ } & {\small UB: $\alpha(n-2)$  } \bigstrut[t]\\

                                    & {\small LB: ~$\alpha$ } & {\small LB:  $\frac{(\alpha+1)(n-2)}{2}$ }  \bigstrut[b]\\
    \hline
    \multirow{2}[2]{*}{\small{Utility}} & {\small UB: 2 } & {\small UB: 2 } \bigstrut[t]\\

                                    &  \multirow{2}[1]{*} {\shortstack{\small LB: 1.5 for $\alpha\ge 3$, \\\small $\min\{\frac{\alpha+1}{2},\frac76\}$ for $\alpha<3$}}
                                    & {\small LB: $1+\frac{\alpha-1}{(2\alpha+2)/t-\alpha}$, for $t=\min\{\frac{1}{3\alpha},1-\frac{1}{\alpha}\}$} \bigstrut[b]\\
    \hline
    \end{tabular}
\end{table}



\vspace{-6mm}
\paragraph{Related work.}

{Existing} preference models in facility location games consider settings where either each agent  cares about its {closest/farthest} facility in the set of preferred facilities \cite{yuan2016facility}
or each agent 
cares about all of the facilities in that set {\cite{serafino2015truthful,serafino2016heterogeneous}.} 
In particular, Yuan \emph{et al.} \cite{yuan2016facility} consider the 2-facility setting where
each agent reports his/her willingness of going to one of the two facilities or to both facilities.
The agent's location in their setting is public, and the agent is interested in his/her {closest/farthest} facility.
Such a setting does not capture the agent's complete ordinal preference and cost/utility of the two facilities directly.
{Serafino and Ventre \cite{serafino2015truthful,serafino2016heterogeneous} consider the setting where each agent reports a subset of preferred facilities, and the cost is the total distance to all preferred facilities.}
Their setting does not necessarily capture agents' ordinal preferences and model agents' interest in going to exactly one facility.

Another line of research on the ordinal preferences is that, every agent reports a linear preference order on the set of candidate facilities to be opened, and a decision-maker opens facilities on the candidates. The optimization problem is considered in \cite{hanjoul1987facility,vasil2009new}. For the mechanism design problem,  Feldman \emph{et al.} \cite{feldman2016voting} study approximation mechanisms in both strategic and non-strategic settings.

There are many studies on different variants of facility location games.   {Feigenbaum \emph{et al.} \cite{feigenbaum2020strategic} consider}
a 1-facility setting
where each agent specifies whether he/she likes or dislikes the facility.
Fong et al. \cite{Fong:aa} consider the 2-facility setting
where each agent reports the fractional preferences (e.g., proportions of usage) of the two facilities. Hossain, Micha, and Shah \cite{hossainsurprising} study the model where each agent may hold several locations on the line with different degrees of importance to the agent, and they introduce a new manipulation: agents may hide some of their locations. 
In the last decade many variants have been studied: strategically reporting the opening costs of facilities \cite{chen2019truthful}, exploring double-peaked preferences of agents \cite{filos2017facility,chen2018mechanism}, different objectives \cite{feldman2013strategyproof} and capacitated constraint \cite{aziz2018capacity,aziz2019facility}. See an overview in \cite{chan2021mechanism}.

\section{Preliminaries}\label{sec:model}

 Let $N=\{1,2,\ldots,n\}$ be the set of agents on a line segment normalized by the unit interval $[0,1]$, 
 and the location profile is $\mathbf x=(x_1,x_2\ldots,x_n)$. 
 For simplicity, assume $x_1\le x_2\le \cdots\le x_n$.
 The planner aims to locate $m$ facilities $F_1,\ldots, F_m$ at some locations on $[0,1]$.
 Each $i\in N$ has an ordinal preference over the $m$ facilities.
 {Agent $i$'s ordinal preference  is denoted by a ranking $\boldsymbol \sigma_i$ over the $m$ facilities.} Denote by $\boldsymbol \sigma=(\boldsymbol \sigma_1,\ldots,\boldsymbol \sigma_n)$ the agents' preference profile. 

 A profile $\mathbf p=(\mathbf x,\boldsymbol \sigma)$ is a collection of the location and preference reported by all agents. A (deterministic) mechanism is a function $f$ which maps profile $\mathbf p$ to an output $\mathbf y=\langle y_1,\ldots,y_m\rangle\in [0,1]^m$ that {locates} facility $F_j$ at $y_j$. 

A mechanism is \emph{strategyproof} (SP), if by reporting the information truthfully,
each agent gains at least as much as that when misreporting, regardless of what others do, under the mechanism's outputs. 
A mechanism is \emph{group strategyproof} (GSP), if no group of agents can collude to misreport their information in a way that makes every member better off.

\smallskip\noindent\textbf{Multiplicative model.} Let $d(a,b)=|a-b|$ be the Euclidean distance between $a$ and $b$. Define $d(x,\mathbf y)=\min_{y\in \mathbf y}d(x,y)$ for any point $x$ and location profile $\mathbf y$.
Let $1=\alpha_1\le\alpha_2\le\cdots\le\alpha_m$ be constant coefficients.
We use these coefficients to characterize the different preference of agents over the facilities.
Each agent incurs a cost equal to the minimum among the distance to his
$k$-th choice multiplied by a factor $\alpha_{k}$.
With some abuse of notations, we assume each agent $i\in N$ has the preference $\boldsymbol \sigma_i=(\sigma_1,\ldots,\sigma_m)$, which indicates that $F_{\sigma_j}$ is the $j$-th most preferred facility. We consider the following objectives with respect to costs and utilities.

\smallskip\emph{Cost objectives.} Given the facilities' location profile $\mathbf y=\langle y_1,\ldots,y_m\rangle$, each $i\in N$ with preference $\boldsymbol \sigma_i$ 
has a cost
$$ c_i(\mathbf y)=\min\{\alpha_1 d(x_i,y_{\sigma_1}),\alpha_2 d(x_i,y_{\sigma_2}),\ldots,\alpha_{m}d(x_i,y_{\sigma_m})\},$$
where $y_{\sigma_j}$ is the location of the $j$-th preferred facility (i.e. $F_{\sigma_j}$) of agent $i$.
 That is, the cost of an agent equals the minimum weighted distance {among the multiplicative weighted distances to all facilities. } 
We wish to minimize the total cost $SC(\mathbf y)=\sum_{i\in N}c_i(\mathbf y)$ or the maximum cost $BC(\mathbf y)=\max_{i\in N}c_i(\mathbf y)$.  
We say a {strategyproof} mechanism $f$ is $r$-approximate with a number $r\ge 1$ under the objective of minimizing the total (resp. maximum) cost, if for any {(truthful)} profile $\mathbf p$, the output satisfies $\frac{SC(f(\mathbf p))}{OPT_S(\mathbf p)}\le r$ (resp. $\frac{BC(f(\mathbf p))}{OPT_B(\mathbf p)}\le r$), where $OPT_S(\mathbf p)$ (resp. $OPT_B(\mathbf p)$) is the optimal objective value of the instance with $\mathbf p$.

\smallskip\emph{Utility objectives.} Given the facilities' location profile $\mathbf y$, each agent $i\in N$ with preference $\boldsymbol \sigma_i$ has a utility
 {\small $$u_i(\mathbf y)=\max\{\frac{1-d(x_i,y_{\sigma_1})}{\alpha_1}, \frac{1-d(x_i,y_{\sigma_2})}{\alpha_2},\ldots,\frac{1-d(x_i,y_{\sigma_m})}{\alpha_m}\}.$$}That is,
 the agent's utility is the maximum discounted utility (based on the ordinal preferences) of the facilities.
 We wish to maximize the total utility $SU(\mathbf y)= \sum_{i\in N}u_i(\mathbf y)$ or minimum utility {$BU(\mathbf y)=\min_{i\in N}u_i(\mathbf y)$.} 
 The approximation ratio can be defined similarly.

Our goal is to design (group) strategyproof mechanisms that (approximately) optimize the objective values.
{All mechanisms considered in this paper are deterministic. In most part of the paper (except Section \ref{sec:ex}), we consider locating two facilities, that is, $m=2$. For notation convenience, we denote the preference of each agent by the index of his preferred facility (e.g., the preference of agent $i$ who prefers $F_1$ is $\boldsymbol \sigma_i=1$), and the preference profile is $\boldsymbol \sigma\in\{1,2\}^n$.  }

When there are two facilities, we write $\alpha_2$ simply as $\alpha$, ignoring $\alpha_1 = 1$.  Given coefficient $\alpha$, denote by $\Gamma^{\alpha}$ the mechanism design problem
 for optimizing a specific system objective, and denote by $\Gamma^{\alpha}(\mathbf p)$ (or simply $\Gamma(\mathbf p)$) an instance {of problem $\Gamma^{\alpha}$} with the specific profile $\mathbf p$. 

 {For example, consider a 3-agent instance with location profile $\mathbf x=(0,0.4,1)$ and $\alpha=3$. Agents 1 and 3 prefer $F_1$, and agent 2 prefers $F_2$. Let $\mathbf y=(0.2,0.8)$ be the facility locations. Then the cost of agent 1 is $c_1(\mathbf y)=\min\{0.2,0.8\alpha\}=0.2$, and $c_2(\mathbf y)=\min\{0.2\alpha,0.4\}=0.4$, $c_3(\mathbf y)=\min\{0.8,0.2\alpha\}=0.6$.
 }

 Before presenting upper and lower bounds on the approximation ratio of strategyproof mechanisms, we give the following result on the relationship between $\Gamma^1$ and $\Gamma^{\alpha}$. Note that $\Gamma^1$ is equivalent to {the typical setting of facility location games, where each agent is indifferent of the facilities.} 

 \begin{proposition}\label{prop}
 Under any of the four possible objectives defined above, if a mechanism $f$ is $\beta$-approximate for $\Gamma^1$, then it is $\beta\alpha$-approximate for $\Gamma^{\alpha}$. 
 \end{proposition}
 \begin{proof}
  Denote by $\mathbf y=\langle y_1,y_2\rangle$ the output of mechanism $f$. Let $OPT^{\alpha}$ ($OPT^1$) be the optimal {value} for $\Gamma^{\alpha}$ ($\Gamma^1$) for the associated objective.
 For minimizing the cost objectives,  we have 
 {\small $$SC(\mathbf y)=\sum_{i\in N}c_i(\mathbf y)\le \alpha\sum_{i\in N}d(x_i,\mathbf y)\le \beta\alpha\cdot OPT^{1}\le \beta\alpha\cdot OPT^{\alpha},$$}
 where the second inequality follows because $f$  is $\beta$-approximate for $\Gamma^1$, and the last inequality follows because the optimum is increasing with $\alpha$.
  Similarly we also have $BC(\mathbf y)\le \beta\alpha\cdot OPT^{\alpha}$.
  For maximizing the utility objectives,  we have
{\small $$SU(\mathbf y)=\sum_{i\in N}u_i(\mathbf y)\ge \frac{\sum\limits_{i\in N}(1-d(x_i,\mathbf y))}{\alpha}\ge\frac{OPT^1}{\beta\alpha}\ge \frac{OPT^{\alpha}}{\beta\alpha},$$}
and similarly $BU(\mathbf y)\ge \frac{OPT^{\alpha}}{\beta\alpha}$.
 \end{proof}

 \medskip  In Section \ref{sec:pre}, we study the case when the agents can only strategically report the private preferences, and the locations are publicly known. In Section \ref{sec:loc}, we study the case where the agents can only strategically report the private locations, and the preferences are publicly known. In Section \ref{sec:ex}, we generalize our results to the setting where both locations and preferences can be misreported, and discuss multiple facilities ($m>2$). Further, we discuss an alternative additive model where an agent's cost/utility for a facility is adjusted by adding/subtracting a constant based on the agent's ordinal preference. All omitted proofs can be found in Supplementary Material.

\section{Unknown Preferences}\label{sec:pre}
In this section,
we consider the setting of unknown preferences and known locations {for $m=2$. That is, the preference information of each agent is private.}
In Section \ref{sec:maxcost} and \ref{sec:totalcost}, we study the objectives of minimizing the maximum cost and total cost of agents, respectively.
In Section \ref{sec:minutility} and \ref{sec:totalutility}, we study the objectives of maximizing the minimum utility and total utility, respectively.

 \subsection{Maximum Cost}\label{sec:maxcost}

Given a location profile $\mathbf x\in \mathbb R^n$, let $lt(\mathbf x)$ and $rt(\mathbf x)$ be the leftmost and rightmost location.
Define $cen(\mathbf x)=(lt(\mathbf x)+rt(\mathbf x))/2$. Let $lb(\mathbf x)=\max\{x_i:x_i\le cen(\mathbf x),i\in N\}$ be the closest location to $cen(\mathbf x)$ {on its left}, and $rb(\mathbf x)=\min\{x_i:x_i\ge cen(\mathbf x),i\in N\}$ be the closest location to $cen(\mathbf x)$ {on its right}. Denote $dist(\mathbf x)=\max\{lb(\mathbf x)-lt(\mathbf x),rt(\mathbf x)-rb(\mathbf x)\}$.
We consider the following mechanism, proposed in \cite{procaccia2009approximate}).

 \begin{mechanism}\label{mec:1}
 Locate $F_1$ at $y_1=\frac{lt(\mathbf x)+lb(\mathbf x)}{2}$, and locate $F_2$ at $y_2=\frac{rt(\mathbf x)+rb(\mathbf x)}{2}$.
 \end{mechanism}

 Consider the problem  $\Gamma^1$ with coefficient $\alpha=1$, i.e., the standard two-facility game, in which the cost of each agent is determined by the closer facility. Procaccia and Tennenholtz \cite{procaccia2009approximate} prove that the optimal maximum cost is at least $\frac{dist(\mathbf x)}{2}$, whereas the cost of each agent induced by Mechanism \ref{mec:1} is at most $\frac{dist(\mathbf x)}{2}$. Therefore, Mechanism \ref{mec:1} is optimal for $\Gamma^1$.

\begin{theorem}
For $\Gamma^{\alpha}$ with private preferences, Mechanism \ref{mec:1} is GSP and $\alpha$-approximate  under the maximum cost objective.
\end{theorem}
\begin{proof}
Since Mechanism \ref{mec:1} achieves the optimal maximum cost for $\Gamma^1$, the approximation ratio $\alpha$ for $\Gamma^{\alpha}$ is established from Proposition \ref{prop}. Note that the outcome of Mechanism \ref{mec:1} 
is independent of the preferences reported by agents. Thus, it is GSP.
\end{proof}

Next, we provide a lower bound on the approximation ratio of SP mechanisms.  We consider an instance with two agents at 0 and $\epsilon$ who prefer $F_1$, and two agents at 1 and $1-\epsilon$ who prefer $F_2$. The number $\epsilon$ is small such that the two agents on the left must be served by facility $F_1$. It can be shown that one of the two agents on the left has incentive to {misreport his true preference and move $F_1$ closer.}

  \medskip\begin{theorem}\label{thm:a1}
For $\Gamma^{\alpha}$ with private preferences, no SP mechanism has an approximation ratio less than {$\max\{1,\frac{\alpha}{2}\}$} under the maximum cost objective.
 \end{theorem}
 \begin{proof}
 We prove the theorem by contradiction. Suppose there is an SP mechanism $f$ with approximation ratio $r<\max\{1,\frac{\alpha}{2}\}$. Consider a 4-agent instance $\Gamma(\mathbf x,\boldsymbol \sigma)$ with location profile $\mathbf x=(0,\epsilon,1-\epsilon,1)$ for a sufficiently small $\epsilon>0$, and {true/private} preference profile $\boldsymbol \sigma=(1,1,2,2)$. The optimal solution is $\langle \frac{\epsilon}{2},1-\frac{\epsilon}{2}\rangle$, and the optimal maximum cost is $\frac{\epsilon}{2}$. 
 Let $\mathbf y=\langle y_1,y_2\rangle$ be the output of $f$. Since the approximation ratio is $r$ and
 $\epsilon$ is sufficiently small, agents 1 and 2 (resp. 3 and 4) must be served by facility $F_1$ (resp. $F_2$).
 We discuss two cases $y_1\le \frac{\epsilon}{2}$ and $y_1> \frac{\epsilon}{2}$.

 \textbf{Case 1.} $y_1\le \frac{\epsilon}{2}$. The cost of agent 2 is $c_2(\mathbf y)\ge \frac{\epsilon}{2}$. Suppose agent 2 misreports his preferred facility as $F_2$, i.e.,  the preference profile becomes $\boldsymbol \sigma'=(1,2,2,2)$. Under the new instance $\Gamma(\mathbf x,\boldsymbol \sigma')$, an optimal solution  is $\langle \frac{\alpha\epsilon}{\alpha+1},\frac{1-\epsilon}{2}\rangle$, and the optimal maximum cost is $\frac{\alpha\epsilon}{\alpha+1}$.  Since the approximation ratio is $r$ and $\epsilon$ is sufficiently small, agent 2 must be served by facility $F_1$ (otherwise, agents 3 and 4 need to be served by $F_1$, and the maximum cost among them is at least $\frac{\alpha\epsilon}{2}$). By the approximation ratio, the cost of any agent under mechanism $f$ should be less than $\frac{\alpha^2\epsilon}{2\alpha+2}$, and the distance of agent 2 to $F_1$
 is less than $\frac{\alpha\epsilon}{2\alpha+2}<\frac{\epsilon}{2}\le c_2(\mathbf y)$.
 Now we look at the original instance $\Gamma(\mathbf x,\boldsymbol \sigma)$. By misreporting his preference, agent 2 can decrease his distance to $F_1$, and thus decrease his cost, which contradicts the strategyproofness.

  \textbf{Case 2.} $y_1> \frac{\epsilon}{2}$. The cost of agent 1 is $c_1(\mathbf y)> \frac{\epsilon}{2}$. Suppose agent 1 misreports his preferred facility as $F_2$, i.e.,  the preference profile becomes $\boldsymbol \sigma''=(2,1,2,2)$. An optimal solution for the new instance $\Gamma(\mathbf x,\boldsymbol \sigma'')$ is $\langle \frac{\epsilon}{\alpha+1},\frac{1-\epsilon}{2}\rangle$, 
  and the optimal maximum cost is $\frac{\alpha\epsilon}{\alpha+1}$. For the same reason in Case 1, agent 1 must be served by facility $F_1$. By the approximation ratio, the distance of agent 1 to $F_1$ in mechanism $f$ should be less than $\frac{\alpha\epsilon}{2\alpha+2}<\frac{\epsilon}{2}<c_1(\mathbf y)$.
 Considering the original instance $\Gamma(\mathbf x,\boldsymbol \sigma)$, agent 1 can decrease his cost by misreporting, which contradicts the strategyproofness.
 \end{proof}

 \subsection{Total Cost}\label{sec:totalcost}

We first consider the problem $\Gamma^1$ {(i.e., $\alpha=1$)} where each agent is indifferent of the two facilities. 
To minimize the total cost in $\Gamma^1$ (e.g., the total distance of agents to their closer facility),
 one can compute an optimal solution  for $\Gamma^1$ in $O(n^2)$ time \cite{procaccia2009approximate}: for $i=2,\ldots,n-1$, denote by $y_{1i}$ the median of $x_1,\ldots,x_i$, and by $y_{2i}$ the median of $x_{i+1},\ldots,x_n$; return {the solution with smallest total cost}, among the $n-1$ solutions $\langle y_{1i},y_{2i}\rangle$.
This mechanism provides an $\alpha$-approximate solution for problem $\Gamma^{\alpha}$ (Proposition \ref{prop}). It is clearly GSP, since its output is based only on the agents' locations and ignores the reported preferences.  

\begin{theorem}
For $\Gamma^{\alpha}$ with private preferences, there exists a GSP mechanism with approximation ratio $\alpha$ under the total cost objective.
\end{theorem}

Next, we prove a lower bound for SP mechanisms, via a similar analysis as in the proof of Theorem \ref{thm:a1}.

\begin{theorem}\label{thm:kk}
For $\Gamma^{\alpha}$ with private preferences, no SP mechanism  has an approximation ratio less than $\max\{1,\frac{\alpha}{4}\}$ under the total cost objective.
 \end{theorem}
 \begin{proof}
 We prove the theorem by contradiction. Suppose there is an SP mechanism $f$ with approximation ratio $r<\max\{1,\frac{\alpha}{4}\}$. Consider a 4-agent instance $\Gamma(\mathbf x,\boldsymbol \sigma)$ with location profile $\mathbf x=(0,\epsilon,1-\epsilon,1)$ for a sufficiently small $\epsilon>0$, and preference profile $\boldsymbol \sigma=(1,1,2,2)$.  The optimal total cost is $2\epsilon$, attained by a solution $\langle 0,1\rangle$. 
 Let $\mathbf y=\langle y_1,y_2\rangle$ be the output of $f$. Since the approximation ratio is $r$ and $\epsilon$ is sufficiently small, agents 1 and 2 (resp. 3 and 4) must be served by facility $F_1$ (resp. $F_2$). 
 We discuss two cases $y_1\le \frac{\epsilon}{2}$ and $y_1> \frac{\epsilon}{2}$.

 \textbf{Case 1.} $y_1\le \frac{\epsilon}{2}$. The cost of agent 2 is $c_2(\mathbf y)\ge \frac{\epsilon}{2}$. Suppose agent 2 misreports his preferred facility as $F_2$, i.e.,  the preference profile becomes $\boldsymbol \sigma'=(1,2,2,2)$. An optimal solution for the new instance $\Gamma(\mathbf x,\boldsymbol \sigma')$ is $\langle \epsilon,1\rangle$, and the optimal total cost is $2\epsilon$.  Since the approximation ratio is $r$ and $\epsilon$ is sufficiently small, agent 2 must be served by facility $F_1$ (otherwise, agents 3 and 4 are served by $F_1$, and the total cost is at least $\alpha\epsilon$, contradicting the approximation ratio). By the approximation ratio, the cost of agent 2 under mechanism $f$ is {strictly} less than $\frac{\alpha\epsilon}{2}$,
 and the distance of agent 2 to $F_1$  is less than $\frac{\epsilon}{2}\le c_2(\mathbf y)$.
 Now we look at the original instance $\Gamma(\mathbf x,\boldsymbol \sigma)$. By misreporting his preference, agent 2 can decrease his distance to $F_1$, and thus decrease his cost, which contradicts the strategyproofness. 

  \textbf{Case 2.} $y_1> \frac{\epsilon}{2}$. The cost of agent 1 is $c_1(\mathbf y)> \frac{\epsilon}{2}$. Suppose agent 1 misreports his preferred facility as $F_2$, i.e.,  the preference profile becomes $\boldsymbol \sigma''=(2,1,2,2)$. An optimal solution for the instance $\Gamma(\mathbf x,\boldsymbol \sigma'')$ is $\langle 0,1\rangle$, and the optimal total cost is $2\epsilon$. Also, agent 1 must be served by facility $F_1$. By the approximation ratio, the cost of agent 1 under mechanism $f$ is less than $\frac{\alpha\epsilon}{2}$, and the distance of agent 1 to $F_1$  is less than $\frac{\epsilon}{2}<c_1(\mathbf y)$.
 Considering the original instance $\Gamma(\mathbf x,\boldsymbol \sigma)$, agent 1 can decrease his cost by misreporting, which contradicts the strategyproofness.
 \end{proof}

 \subsection{Minimum Utility} \label{sec:minutility}
In Section \ref{sec:maxcost} we have shown that Mechanism \ref{mec:1} is GSP and achieves the optimal maximum cost for the problem $\Gamma^1$. It is easy to see that it also achieves the optimal minimum utility for $\Gamma^1$. By Proposition \ref{prop}, it is an $\alpha$-approximate mechanism for $\Gamma^{\alpha}$.
Below, we present the following mechanism that locates facility $F_i$ at the midpoint of the locations of {the leftmost and rightmost} agents who prefer $F_i$, for $i=1,2$,
and obtains the optimal minimum utility for $\alpha \ge 2$.


\begin{mechanism}\label{mec:2}
Let $\mathbf x_1=(x_i)_{i\in N:\boldsymbol \sigma_i=1}$ and $\mathbf x_2=(x_i)_{i\in N:\boldsymbol \sigma_i=2}$ be the location profile of agents who prefer $F_1$ and $F_2$, respectively.  Return the solution
$$\langle \frac{lt(\mathbf x_1)+rt(\mathbf x_1)}{2},\frac{lt(\mathbf x_2)+rt(\mathbf x_2)}{2}\rangle.$$
 \end{mechanism}

Next we analyze its performance.

 \begin{theorem}
 For $\Gamma^{\alpha}$ with private preferences, Mechanism \ref{mec:1} is GSP and $\alpha$-approximate under the minimum utility objective. When $\alpha\ge 2$, Mechanism \ref{mec:2} is an optimal SP mechanism.
 \end{theorem}
 \begin{proof}
The approximation ratio of Mechanism \ref{mec:1} is given by Proposition \ref{prop}. We only consider the case when $\alpha\ge 2$.
We first show that Mechanism 2 {(denoted by $f$)} is SP. Consider agent $i\in N$, and assume w.l.o.g. that he prefers $F_1$.
Under the location profile $\mathbf y=\langle y_1,y_2\rangle$ induced by truthful reporting, $i$'s distance to $F_1$ is  $d(x_i,y_1)\le\frac12$,
and $i$'s utility is $\max\{1-d(x_i,y_1),\frac{1-d(x_i,y_2)}{\alpha}\}=1-d(x_i,y_1)\ge \frac12$.
If agent $i$ misreports $F_2$, $F_1$ cannot be closer to $i$, and the only possibility of improvement is that $i$'s utility is determined by $F_2$.
However, in that case, his utility is at most $\frac{1}{\alpha}\le\frac12$. Therefore, agent $i$ cannot benefit by lying.   

{For the optimality, under the outcome of $f$, each agent has a utility at least $\frac12$, and does not want to be served by the less preferred facility as $\alpha\ge2$. So the midpoint $\frac{lt(\mathbf x_1)+rt(\mathbf x_1)}{2}$ (or $\frac{lt(\mathbf x_2)+rt(\mathbf x_2)}{2}$) is the best for agents who prefer $F_1$ (or $F_2$), otherwise some agent's utility will decrease.}
 \end{proof}

Mechanism \ref{mec:2} is no longer SP when $\alpha<2$,
as the agent may have a utility determined by the less preferred facility, and yield a closer distance to it by misreporting.
To complement our result above, we derive the following lower bound. 

 \begin{theorem}\label{thm:36}
For $\Gamma^{\alpha}$ with private preferences, when $1<\alpha<2$, no SP mechanism has an approximation ratio better than $\min\{\frac{2}{\alpha},\frac{3\alpha+1}{2\alpha+2}\}$ under the minimum utility objective.
 \end{theorem}
 \begin{proof}
 When $\alpha<2$, suppose, for contradiction, that $f$ is an SP mechanism with an approximation ratio $r<\min\{\frac{2}{\alpha},\frac{3\alpha+1}{2\alpha+2}\}$.
 Consider an instance $\Gamma(\mathbf x,\boldsymbol \sigma)$ with $\mathbf x=(0,0,x=1-\frac{1}{\alpha},1)$ and $\boldsymbol \sigma=(1,1,1,2)$. 
 Set $x=1-\frac{1}{\alpha}<\frac12$. It is easy to see that, since the left three agents prefer $F_1$, an
optimal solution must locate $F_1$ at $y_1=\frac{x}{2}$, and has a minimum utility of $1-\frac{x}{2}$.
 Let $\mathbf y=\langle y_1,y_2\rangle$ be the output of $f$. By the approximation ratio, $f$ induces a minimum utility at least $\frac{1-x/2}{r}$. If one of the left three agents has his/her utility
determined by $F_2$, then his/her utility is at most $\frac{1}{\alpha}<\frac{1-x/2}{r}$, a contradiction.
 So the utilities of the left three agents should
be determined by $F_1$. It indicates that $y_1\ge \frac{1-x/2}{r}+x-1$, since the distance between agent 3 and $F_1$ should be at most $1-\frac{1-x/2}{r}$.

 Consider another instance $\Gamma(\mathbf x,\boldsymbol \sigma')$ with preference profile $\boldsymbol \sigma=(2,1,1,2)$.  An optimal solution $\langle0,1\rangle$ has a minimum utility of $\frac{1}{\alpha}$. Mechanism $f$ induces a minimum utility at least $\frac{1}{\alpha \cdot r}$.  If the utility of agent 4 is determined by F1, then
the utilities of agents 2 and 3 are determined by F2, and the
minimum utility of them is at most $\frac{1-x/2}{\alpha}<\frac{1}{\alpha\cdot r}$, a contradiction.
So the utility of agent 4 is determined by $F_2$. 
 Let $\mathbf y'=\langle y_1',y_2'\rangle$ be the output of $f$. We discuss the following two cases with respect to the utility of agent 1 under $\mathbf y'$.

 \noindent\textbf{Case 1.} \emph{Agent 1's utility is determined by $F_2$.} To guarantee the approximation ratio, $F_2$ needs to serve agent 1, 2 and 3.
 Then the minimum utility of agents 2 and 3 induced by $f$ is at most  $\frac{1-x/2}{\alpha}<\frac{1}{\alpha\cdot r}$, a contradiction with the approximation ratio $r$.

 \noindent\textbf{Case 2.} \emph{Agent 1's utility is determined by $F_1$.}
 Agent 1's utility is $\frac{1-y_1'}{\alpha}$. To satisfy $\frac{1-y_1'}{\alpha}\ge \frac{1}{\alpha\cdot r}$, it must have $y_1'\le 1-\frac1r$.

 Now we move back to the original instance $\Gamma(\mathbf x,\boldsymbol \sigma)$. If agent 1 truthfully report $F_1$  as his/her top choice, mechanism $f$ locates $F_1$ at $y_1\ge \frac{1-x/2}{r}+x-1$. However, if he/she misreports $F_2$, $f$ would locate $F_1$ at $y'_1\le 1- \frac1r <\frac{1-x/2}{r}+x-1$, reducing the distance between agent 1 and $F_1$. Hence, agent 1 has incentive to lie, and $f$ is not SP.
 \end{proof}

 \subsection{Total Utility} \label{sec:totalutility}
For maximizing the total utility of the agents, we present the following fixed mechanism that always locates the two facilities at the midpoint of the interval [0,1].

  \begin{mechanism}\label{mec:3}
  Locate $F_1$ and $F_2$ both at $\frac12$.
  \end{mechanism}

 \begin{theorem}\label{thm:with}
For $\Gamma^{\alpha}$ with private preferences, Mechanism \ref{mec:3} is GSP and 2-approximate under the total utility objective.
 \end{theorem}
 \begin{proof}
 The mechanism is trivially GSP. Notice that each agent incurs a distance at most $\frac12$ to his preferred facility under this mechanism, and thus has a utility at least $\frac12$. It induces a total utility at least $\frac n2$, whereas the optimal total utility is no more than $n$, giving an approximation ratio 2.
 \end{proof}

{Note that the optimal solutions for $\Gamma^1$ under the total cost and total utility objectives are the same. By Proposition \ref{prop}, the optimal mechanism for $\Gamma^1$ (stated in Section \ref{sec:totalcost}) has an approximation ratio $\alpha$ for $\Gamma^{\alpha}$. So we can improve Theorem \ref{thm:with} when $\alpha< 2$.}

 \begin{corollary}
For $\Gamma^{\alpha}$ with private preferences, there exists a GSP mechanism with approximation ratio $\min\{2,\alpha\}$ under the total utility objective.
 \end{corollary}

 Next, we derive a lower bound for the total utility setting. 
 {It is equal to 1 for $\alpha=1$ and approaches 1 when $\alpha$ grows.
}

 \begin{theorem}\label{thm:39}
For $\Gamma^{\alpha}$ with private preferences,  no SP mechanism has an approximation ratio better than $\frac{30\alpha^2+38\alpha+2}{30\alpha^2+37\alpha+3}$  under the total utility objective.
  \end{theorem}
 \begin{proof}
  Suppose $f$ is an SP mechanism with approximation ratio $r<\frac{30\alpha^2+38\alpha+2}{30\alpha^2+37\alpha+3}$.
  Consider an instance $\Gamma(\mathbf x,\boldsymbol \sigma)$ with $\mathbf x=(0,s,2s,3s,1,1)$ and $\boldsymbol\sigma=(2,1,1,1,2,2)$
  where $s=\frac{1}{3(\alpha+1)}$. {An optimal solution should locate $F_1$ near the left four agents, and locate $F_2$ near the right two agents. Indeed,
   $\langle2s,1\rangle$ is optimal, and} has a  total utility of $OPT=5-2s+\frac{1-2s}{\alpha}$.
  Let $\mathbf y=\langle y_1,y_2\rangle$ be the output of $f$. We claim that $y_1>\frac{3s}{2}$: if $y_1\le\frac{3s}{2}$, {by an overestimation}, the utility induced by $f$ is less than $5-2s-\frac{s}{2}+\frac{1}{\alpha}-\frac{3s}{2\alpha}<\frac{OPT}{r}$, a contradiction with the approximation ratio.

   For another instance $\Gamma(\mathbf x,\boldsymbol \sigma')$  with $\boldsymbol \sigma'=(1,1,1,2,2,2)$, let $\mathbf y'=\langle y'_1,y'_2\rangle$ be the output of $f$. By a symmetric analysis we have $y'_1< \frac{3s}{2}$.

  Now, we consider instance $\Gamma(\mathbf x,\boldsymbol \sigma'')$ with  $\boldsymbol \sigma''=(2,1,1,2,2,2)$. The optimal solution is to locate $y_1\in[s,2s]$ and $y_2=1$. 
 Let $\mathbf y''=\langle y''_1,y''_2\rangle$ be the output of $f$. By the approximation ratio, $f$ must locate $F_1$ close to the interval $[s,2s]$,  and $F_2$ close to 1.
   \begin{itemize}
   \item If  $y''_1\le \frac{3s}{2}$, then agent 4 can misreport his top choice as $F_1$, because in that case $f$  locates $F_1$ at $y_1>\frac{3s}{2}$. Then agent 4 becomes closer to $F_1$ and has a larger utility.
   \item If $y''_1\ge \frac{3s}{2}$, then agent 1 can misreport his top choice as $F_1$, because in that case $f$  locates $F_1$ at $y_1'<\frac{3s}{2}$. Then agent 1 becomes closer to $F_1$ and has a larger utility.
       \end{itemize}
       Hence, $f$ cannot be SP and $r$-approximate for $\Gamma(\mathbf x,\boldsymbol \sigma'')$.
  \end{proof}


\section{Unknown Locations}\label{sec:loc}

In this section, we study settings where the location profile of the agents is private and the preference profile is public. {For $m=2$, we} study the objectives of minimizing the total and maximum costs as well as maximizing the total and minimum utilities in Section \ref{sec:loc_cost} and Section \ref{sec:loc_utility}, respectively.

\subsection{Minimizing the Maximum and Total Costs} \label{sec:loc_cost}
 Procaccia and Tennenholtz \cite{procaccia2009approximate} propose the following  GSP mechanism for the standard two-facility game (i.e., $\Gamma^1$), which locates two facilities at the leftmost and rightmost locations of agents. It has an approximation ratio 2 and $(n-2)$ for minimizing the maximum cost and total cost, respectively. By Proposition \ref{prop}, it has a guarantee for $\Gamma^{\alpha}$.

\begin{mechanism}\label{mec:4}
Locate $F_1$ and $F_2$ at $\langle lt(\mathbf x),rt(\mathbf x)\rangle$.
\end{mechanism}


\begin{theorem}
For $\Gamma^{\alpha}$ with private locations, Mechanism \ref{mec:4} is  GSP and $2\alpha$-approximate (resp. $\alpha(n-2)$-approximate) under the maximum (resp. total) cost objective.
\end{theorem}

Fotakis and Tzamos \cite{fotakis2014power} prove that, for any deterministic anonymous\footnote{{A mechanism is anonymous if permuting the agents does not change the solution.}} SP mechanism $f$ with a bounded approximation ratio for the standard 2-facility location game $\Gamma^1$, and for {any instance with all agents located at 3 locations,} 
mechanism $f$ places the facilities at the two extreme locations $\langle lt(\mathbf x),rt(\mathbf x)\rangle$.
We can show that this property also holds for $\Gamma^{\alpha}$, and use it to prove lower bounds.

\begin{lemma}\label{lem:ext}
For any anonymous SP mechanism $f$ with a bounded approximation ratio for $\Gamma^{\alpha}$ with both objectives, and for any instance with all agents located at 3 locations, $f$ places the facilities at the two extreme locations $\langle lt(\mathbf x),rt(\mathbf x)\rangle$.
\end{lemma}
\begin{proof}
  Let $f$ be an arbitrary deterministic anonymous SP mechanism with a bounded approximation ratio for $\Gamma(\alpha)$ under the total cost objective. Then it also has a bounded ratio under the maximum cost objective. Let $\Gamma(\mathbf x,\boldsymbol \sigma)$ be an arbitrary
instance of $\Gamma(\alpha)$  with $n = 3$ agents. Let $\mathbf p=(\mathbf x,\boldsymbol \sigma)$ be the profile, and $f(\mathbf p) = \langle y_1,y_2\rangle$ be the output. Assume w.l.o.g. $x_1\le x_2\le x_3$. Using the same proof for Lemma 3.1 of \cite{fotakis2014power}, we can prove that the rightmost
facility is always located to the rightmost agent, that is, $y_1 =
x_3$ if $y_1\ge y_2$, and $y_2=x_3$  otherwise. Using a symmetric
argument, we can show that the leftmost facility is always
located to the leftmost agent, that is, $y_1=x_1$ if $y_1\le y_2$, and $y_2=x_1$ otherwise. These two claims imply that $f$ must place the facilities at the two extreme locations.

Now we consider an instance $\Gamma(\mathbf x',\boldsymbol \sigma')$ where the $n\ge3$ agents are located at 3 locations. We call it a 3-location instance. We can extend the above conclusion for 3-agent instances to 3-location instances by restating the proofs with
three coalitions of agents instead of three agents. A mechanism is partial GSP if for any group of agents at the same
location, each individual cannot benefit if they misreport simultaneously. Using the fact that any SP mechanism is also
partial GSP (Lemma 2.1 of \cite{lu2010asymptotically}), we obtain that $f$ must place the facilities at the two extreme locations.
  \end{proof}

\begin{theorem}
{Under the maximum cost (resp. total cost) objective, no {anonymous} SP mechanism for $\Gamma^{\alpha}$ with private locations has an approximation ratio less than $\alpha$ (resp. $\frac{(\alpha+1)(n-2)}{2}$).}
\end{theorem}
\begin{proof}
 Suppose that there is an {anonymous} SP mechanism $f$. Consider a 3-location instance $\Gamma(\mathbf x,\boldsymbol \sigma)$ with location profile $\mathbf x=(0,\epsilon,\ldots,\epsilon,1)$ for a small $\epsilon>0$, and preference profile $\boldsymbol \sigma=(1,\ldots,1,2,\ldots,2)$ with $\frac{n}{2}$ agents preferring $F_1$.

 Under the maximum cost objective, the optimal maximum cost is at most $\epsilon$ (attained by a solution $\langle \epsilon,1\rangle$). By Lemma \ref{lem:ext}, $f$ places the facilities at 0 and 1, and the maximum cost  is $\alpha\epsilon$. So $f$ cannot have an approximation ratio less than $\alpha$. 

 Under the total cost objective, the optimal total cost is $\epsilon$, attained by a solution $\langle \epsilon,1\rangle$. Mechanism $f$ places the facilities at 0 and 1, and the total cost is $\frac{\alpha(n-2)\epsilon}{2}+\frac{(n-2)\epsilon}{2}$.  So $f$ cannot have an approximation ratio less than $\frac{(\alpha+1)(n-2)}{2}$.
\end{proof}

For minimizing the total cost, an intuitive idea is to locate $F_1$ and $F_2$ at the median of the locations of agents who prefer $F_1$ and $F_2$, respectively. 
Since any misreporting from agent $i\in N$ would not change the location of the less preferred facility or reduce the distance to the preferred facility, this mechanism is SP. 
However, it cannot achieve a bounded approximation ratio. The following example illustrates this.

\begin{example}
Consider the location profile $\mathbf x=(0,0,0,1)$, and the preference profile $\boldsymbol \sigma=(1,2,2,2)$. The mechanism would locate both $F_1$ and $F_2$ at 0.  The induced total cost is 1, whereas the optimal solution $\langle0,1\rangle$ has a total cost of 0.
\end{example}

\subsection{Maximizing the Minimum and Total Utilities} \label{sec:loc_utility}
Note that the utility of an agent is at most 1. Mechanism \ref{mec:3}, which locates two facilities at a fixed point $\frac12$, induces a utility at least $\frac12$ for each agent. We have the following.
 \begin{theorem}
For $\Gamma^{\alpha}$ with private locations, Mechanism \ref{mec:3} is GSP and 2-approximate for both objectives of maximizing the total utility and the minimum utility.
 \end{theorem}

Next, we provide lower bounds for both of the objectives.
 \begin{theorem}
  Under the minimum utility objective, no SP mechanism for the problem $\Gamma^{\alpha}$ with private locations has an approximation ratio better than 1.5 when $\alpha\ge 3$.
 \end{theorem}
 \begin{proof}
 When $\alpha\ge 3$, suppose that $f$ is an SP mechanism with approximation ratio $r<1.5$.
 Consider an instance $\Gamma(\mathbf x,\boldsymbol \sigma)$  with $\mathbf x=(0,\frac12,1)$ and $\boldsymbol \sigma=(1,2,1)$.
  Clearly, the optimal minimum utility is $\frac12$, with respect to the solution $\langle\frac12,\frac12\rangle$. So the minimum utility induced by $f$ is larger than $\frac13$. Let $\mathbf y=\langle y_1,y_2\rangle$ be the output. If $y_1<\frac13$, then the utility of agent 3 is $\max\{y_1,\frac{y_2}{\alpha}\}\le \frac13$, a contradiction. Symmetrically, it cannot be $y_1>\frac23$. Then it must be $y_1\in(\frac13,\frac23)$. Assume w.l.o.g.  that $y_1\in (\frac13,\frac12]$. 

 Now we consider another instance $\Gamma(\mathbf x',\boldsymbol \sigma)$ with location profile $\mathbf x'=(y_1,\frac12,1)$. Note that the utility of an agent who is served by the less preferred facility is at most $\frac{1}{\alpha}\le\frac13$, and thus a good solution should make each agent be served by the preferred one. Thus, solution
$\langle \frac{1+y_1}{2},\frac12\rangle$ is optimal, and has a minimum utility of $\frac{1+y_1}{2}$.
 Let $\mathbf y'=\langle y'_1,y'_2\rangle$ be the output of $f$.
 By the approximation ratio, $f$ must locate $F_1$ at $y'_1>y_1$ (otherwise, the utility of agent 3 is $1-(1-y'_1)=y'_1\le y_1<\frac{1+y_1}{2r}$).
 {Then, if agent 1 misreports the location as 0, the mechanism receives a profile $(\mathbf x, \boldsymbol \sigma)$, and would locate $F_1$ at $y_1$ by the previous analysis. Hence, under instance $\Gamma(\mathbf x',\boldsymbol \sigma)$, agent 1 would gain by misreporting his location as 0 (receiving a utility of 1), a contradiction.}
\end{proof}

\begin{theorem}\label{thm:47}
  Under the minimum utility objective, no SP mechanism for  $\Gamma^{\alpha}$ with private locations has an approximation ratio better than {$\min\{\frac{\alpha+1}{2},\frac76\}$ when $1\le \alpha<3$.}
 \end{theorem}
  \begin{proof}
When $1\le \alpha<3$, suppose that $f$ is an SP mechanism with approximation ratio $r<\min\{\frac{\alpha+1}{2},\frac76\}\le \alpha$.
Consider  the agents' location profile $\mathbf{x}=(0,\frac12,1)$ {and preference profile $\boldsymbol{\sigma}=(1,1,2)$.} An optimal solution $\langle\frac14,1\rangle$ has a minimum utility of $\frac34$. Let $\langle y_1,y_2\rangle$ be the output. If $y_2\le y_1$, {it is easy to see that facility $F_2$ cannot serve agents 1 and 2 simultaneously, otherwise one of them has a utility at most $\frac{3}{4\alpha}$, violating the approximation ratio. Also, $F_2$ cannot serve agent 3, otherwise agent 1 would have an unacceptably small utility.  So agents 2 and 3 must be served by $F_1$. However, for all possible locations of $F_1$, at least one of agents 2 and 3 has a utility at most $\frac{3}{2(\alpha+1)}$, contradicting the approximation ratio.  }

 Next, we consider the case $y_1< y_2$, where agent 1 is served by $F_1$. If $y_1\ge\frac12$, then the utility of agent 1 is at most $\frac12$, violating the approximation ratio. So it must be $y_1<\frac12$. We discuss two cases.

\noindent\textbf{Case 1.} $y_1\ge \frac14$. Consider the location profile $\mathbf{x'}=(0,y_1,1)$, and let $\langle y_1',y_2'\rangle$ be the output of $f$.  The optimal minimum utility is $1-\frac{y_1}{2}$, attained by the solution $\langle \frac{y_1}{2},1\rangle$. It cannot be $y_1'=y_1$, otherwise the minimum utility is at most $1-y_1<(1-\frac{y_1}{2})\frac67$, contradicting the approximation ratio.
Thus, {under instance $\Gamma^{\alpha}(\mathbf{x'},\boldsymbol{\sigma})$,} agent 2 can misreport  $\frac12$, such that {the mechanism, which receives a location profile $\mathbf{x}$, locates facility $F_1$} at his true location $y_1$, increasing his utility to 1.

\noindent\textbf{Case 2.} $y_1< \frac14$. Consider the location profile $\mathbf{x}''=(y_1,\frac12,1)$, and let $\langle y_1'',y_2''\rangle$ be the output of $f$. The optimal minimum utility is $\frac34+\frac{y_1}{2}$. Similarly, by the approximation ratio, it cannot be $y_1''=y_1$. Thus, {under instance $\Gamma^{\alpha}(\mathbf{x''},\boldsymbol{\sigma})$,} agent 1 can misreport his location as 0 such that {the mechanism, which receives a location profile $\mathbf{x}$, locates facility $F_1$} at his true location $y_1$, increasing his utility to 1.
\end{proof}

 \begin{theorem}\label{thm:48}
  Under the total utility objective, every SP mechanism for the problem $\Gamma^{\alpha}$ with private locations has an approximation ratio at least $1+\frac{\alpha-1}{(2\alpha+2)/t-\alpha}$, where $t=\min\{\frac{1}{3\alpha},1-\frac{1}{\alpha}\}$.
\end{theorem}
\begin{proof}
Suppose that $f$ is an SP mechanism with approximation ratio $r<1+\frac{\alpha-1}{(2\alpha+2)/t-\alpha}$. {Define $t':=t-\epsilon<\frac{1}{3\alpha}$} with a sufficiently small $\epsilon>0$.
Consider a 4-agent instance with location profile $\mathbf x=(0,t',1-t',1)$ and preference profile $(1,1,1,1)$.
Solution $\langle 0,1\rangle$ is optimal, and has a total utility of $2-t'+\frac{2-t'}{\alpha}$. 
Let $\mathbf y=\langle y_1,y_2\rangle$ be the output of $f$, and assume w.l.o.g. that $y_1\le y_2$. By the approximation ratio, we have $y_1\le2t'$. Then the utility of agent 3 is at most $\max\{3t',\frac{1}{\alpha}\}=\frac{1}{\alpha}$. {Actually, the value of $t'$ and the approximation ratio guarantee that, if $y_1\le y_2$, the rightmost two agents are served by $F_2$, and vice versa.}

Now consider agent 3 misreporting his location as 1, and the location profile becomes $\mathbf x'=(0,t',1,1)$. The optimal total utility for this new instance is $2+\frac{2-t'}{\alpha}$, attained by the solution $\langle 1,0\rangle$. Let $\mathbf y'=\langle y'_1,y'_2\rangle$ be the output of $f$. If $y_1'\le y_2'$, then the best possible total utility is $2-t'+\frac{2}{\alpha}$, attained by the solution $\langle 0,1\rangle$. By the approximation ratio, it is impossible, and thus it must be $y_1'>y_2'$. Also, by the approximation ratio, we have $y_1'\ge 1-2t'$. Therefore, {after misreporting,} the utility of agent 3 (whose true location is $1-t'$) becomes at least $1-t'>\frac{1}{\alpha}$, implying that he gains by misreporting. Therefore, the mechanism is not SP.
\end{proof}


\vspace{-3mm}\section{Extensions}\label{sec:ex}

\subsection{Private Locations and Preferences}
We can extend our results to the setting of both private location and preference of agents. 
It is a generalization of the setting of either private locations or private preferences. Hence, all lower bound results in Section \ref{sec:pre} and \ref{sec:loc}  are applicable to this setting. For  the upper bound, we notice that Mechanism \ref{mec:4} (which locates two facilities at the two extreme locations of agents, and is $2\alpha$-approximate and $(n-2)\alpha$-approximate for the objectives of minimizing the maximum cost and the total cost) is still GSP, because any misreport on preferences would not affect the output, and any misreport on locations cannot make some agent be closer to the facilities. Mechanism \ref{mec:3} (which locates two facilities both at the point $\frac12$, and is 2-approximate for both objectives of maximizing the minimum utility and the total utility) is clearly GSP, because it does not rely on any reports and has a fixed output.

We can generalize our results to more than two facilities: suppose there are $m$ facilities $F_1,\ldots,F_m$ to be built, and each agent reports the private profile (including an ordinal preference over $m$ facilities) to the planner. Each agent incurs a cost equal to the minimum among  the distance to his $i$-th choice multiplied by a factor $\alpha_{i}$. 
 Proposition \ref{prop} can be restated as: a mechanism is $\beta\alpha_{m}$-approximate, if it is $\beta$-approximate for the typical setting of $m$-facility-location games. Based on that, the results in Section \ref{sec:pre} can be generalized as: for both objectives of minimizing the maximum cost and total cost, we have a GSP and $\alpha_{m}$-approximate mechanism when misreporting only preferences. The utility version can be generalized similarly: Mechanism \ref{mec:3} is GSP and 2-approximate for the minimum/total utility objective, {even without the restriction on misreporting.}

\subsection{Additive Variations}

{Alternatively, we can also} consider a setting where the cost of an agent served by a less preferred facility
is \emph{added} by a factor as presented below, {instead of the multiplicative adjustment.}

 \medskip\noindent\textbf{Cost objectives.} Given the facilities' location profile $\mathbf y=\langle y_1,\ldots,y_m\rangle$, each agent $i\in N$ with preference $\boldsymbol \sigma_i=(\sigma_1,\ldots,\sigma_n)$ has a cost
\begin{align*}
 c_i(\mathbf y)=\min\{& d(x_i,y_{\sigma_1})+\alpha_1, d(x_i,y_{\sigma_2})+\alpha_2,\\
 & \ldots,d(x_i,y_{\sigma_m})+\alpha_m\}
 \end{align*}
with coefficient {$0=\alpha_1\le \alpha_2\le\cdots\le\alpha_m\le 1$.}
We wish to minimize the total cost $SC(\mathbf y)=\sum_{i\in N}c_i(\mathbf y)$ or the maximum cost $BC(\mathbf y)=\max_{i\in N}c_i(\mathbf y)$.

\medskip\noindent\textbf{Utility objectives.} Given the facilities' location profile $\mathbf y$, each agent $i\in N$ with preference $\boldsymbol \sigma_i$ has a utility
 \begin{align*}
 u_i(\mathbf y)=\max\{& 1-d(x_i,y_{\sigma_1})-\alpha_1, 1-d(x_i,y_{\sigma_2})-\alpha_2,\\
 & \ldots,1-d(x_i,y_{\sigma_m})-\alpha_m\}
 \end{align*}
 with coefficient {$0=\alpha_1\le \alpha_2\le\cdots\le\alpha_m\le 1$.}
 We wish to maximize the total utility $SU(\mathbf y)= \sum_{i\in N}u_i(\mathbf y)$ or minimum utility {$BU(\mathbf y)=\min_{i\in N}u_i(\mathbf y)$.}

In contrast to the multiplicative version, Proposition \ref{prop} no longer applies, and we cannot directly utilize the results on the standard setting.
We present preliminary results for the case of two facility under different private information.

\smallskip\noindent\textbf{Unknown preferences.}
We show that any mechanism that does not consider preference must have unbounded approximation ratio for the social/maximum cost objective.
Consider an instance where all agents preferring $F_1$ are located at 0 and all agents preferring $F_2$ are located at 1.
The optimal social/maximum cost is 0, while a mechanism ignoring the preferences may induce a positive social/maximum cost.

{Though designing good mechanisms for the cost objectives is open, we can obtain better results for the utility objectives. It is easy to see that, for the total/minimum utility objective, Mechanism \ref{mec:3} is GSP and 2-approximate. In particular, when $\alpha_2\ge 0.5$, Mechanism \ref{mec:2} is GSP and optimal. }

\smallskip\noindent\textbf{Unknown locations.}
{Any mechanism ignoring the preferences would also have unbounded approximation ratio for the cost objectives, including Mechanism \ref{mec:4} which selects the two extreme locations of agents. Another intuitive approach is locating $F_1$ at the median of agents who prefer $F_1$, and locating $F_2$ at the median of agents who prefer $F_2$. 
It guarantees the strategyproofness, but has a bad approximation ratio when $\alpha$ is small. Under the total/minimum utility objective, Mechanism \ref{mec:3} is clearly GSP and 2-approximate. Mechanism \ref{mec:2} is no longer SP, because an agent can manipulate the location such that the preferred facility is built on his true location. }




\section{Conclusion}\label{sec:dis}

For the facility location problem where each agent has an ordinal preference over facilities, we derive upper and lower bounds on the approximation ratio of truthful mechanisms under four possible objectives, {which are asymptotically tight up to a constant.} It remains an interesting open problem to narrow the gaps of the bounds. {Randomization would be an attractive direction.} The randomized Proportional Mechanism proposed in \cite{lu2010asymptotically} (which locates $F_1$ on the location of a randomly chosen agent, and locates $F_2$ at the location of agent $i\in N$ with probability proportional to his (weighted) distance to $F_1$) is SP and 4-approximate for the standard 2-facility game under the total cost objective; but it {cannot be used} 
for our model with preferences, because it is not SP  for both settings of unknown preferences and unknown locations. The Random-Dictatorship-like mechanisms can also have unbounded approximation ratios. In view of these, we believe that the lower bounds can be proved to be greater via more carefully designed instances. 
{Also, it would be interesting to further explore the additive variations.} 

%
%
%

\bibliographystyle{splncs04}
\bibliography{reference}

\end{document}